\documentclass{article}

\usepackage{amsthm}
\usepackage{graphicx}

\author{\addtocounter{footnote}{1}  
Jean-Guillaume Dumas\footnote{
Universit\'e de Grenoble. 
Laboratoire LJK,
umr CNRS, INRIA, UJF, UPMF, GINP.
51, av. des Math\'ematiques, F38041 Grenoble, France.
\href{mailto:Jean-Guillaume.Dumas@imag.fr}{Jean-Guillaume.Dumas@imag.fr},
\href{http://ljk.imag.fr/membres/Jean-Guillaume.Dumas/}{ljk.imag.fr/membres/Jean-Guillaume.Dumas/}
}
\and  
Erich Kaltofen\footnote{
Department of Mathematics.
North Carolina State University.
Raleigh, NC 27695-8205, USA.
\href{mailto:kaltofen@math.ncsu.edu}{kaltofen@math.ncsu.edu},
\href{http://www.kaltofen.us}{www.kaltofen.us}
}
}

\newcommand{\customlinebreak}{}
\newcommand{\strechparskip}[1]{}
\newcommand{\strechparsep}[1]{}
\newcommand{\customvspace}[1]{}
\newcommand{\category}[3]{}
\newcommand{\terms}[1]{}
\newcommand{\keywords}[1]{}
\newcommand{\balancecolumns}{}
\newcommand{\makeconference}{}
\let\oldresizebox\resizebox
\renewcommand{\resizebox}[3]{\oldresizebox{\linewidth}{!}{#3}}

\usepackage{amsmath,amssymb,amsfonts,xspace}
\usepackage[all]{xy}

\newtheorem{theorem}{Theorem}
\newtheorem{corollary}{Corollary}
\newtheorem{lemma}{Lemma}

\newtheorem{remark}{Remark}

\newcommand{\F}{\ensuremath{\mathbb F}}
\newcommand{\B}{\ensuremath{\mathbb B}}
\newcommand{\Z}{\ensuremath{\mathbb Z}}\newcommand{\ZZ}{\Z}
 
\newcommand{\Zp}{\leavevmode\kern.1em\raise.0ex
  \hbox{\ensuremath{\mathbb{Z}}}\kern-.1em
  /\kern-.15em\lower.3ex\hbox{p\mbox{\ensuremath{\mathbb{Z}}}}\xspace}
\newcommand{\R}{\ensuremath{\mathbb R}}
\newcommand{\W}{\ensuremath{\mathbb W}}
 
\newcommand{\bigO}[1]{\ensuremath{O(#1)}\xspace}  
\newcommand{\card}[1]{\ensuremath{\left|#1\right|}\xspace}

\newcommand{\sdp}{semidefinite program\xspace}

\newcommand{\psd}{positive semidefinite\xspace}

\newcommand{\psdness}{positive semidefiniteness\xspace}

\newcommand{\symm}[2]{\mathbb{S}#1^{#2\times #2}}

\newcommand{\sos}{sum-of-squares\xspace}

\newcommand{\lognormW}{\log\| W\mspace{1mu} \|}
\newcommand{\lognormA}{\log\| A\mspace{1mu} \|}

\title{Essentially Optimal Interactive Certificates in Linear Algebra\texorpdfstring{\raisebox{0ex}{*}}{}} 

\makeatletter
\usepackage{color,svgcolor}
\usepackage[ 
plainpages=true 
,bookmarks=false 
]{hyperref}
\hypersetup{
  pdftitle={\@title},
  pdfauthor={Jean-Guillaume Dumas and Erich Kaltofen},
  breaklinks=true, 
  colorlinks=true,
  linkcolor=darkred,
  citecolor=blue,
  urlcolor=darkgreen,   
}
\makeatother
\usepackage[all]{hypcap}

\begin{document}
\makeconference
\maketitle
\def\thefootnote{\fnsymbol{footnote}}  
\footnotetext[1]{ 
\scriptsize
This research was supported in part by    
the Agence Nationale pour la Recherche    
under Grant ANR-11-BS02-013 HPAC (Dumas)  
and                                       
the National Science Foundation
under Grant CCF-1115772 (Kaltofen).
}

\begin{abstract}

Certificates to a linear algebra computation are
additional data structures for each output, which can be used
by a---possibly randomized---verification algorithm that proves the
correctness of each output.  The certificates are essentially optimal
if the time
(and space)  
complexity of verification is essentially linear
in the input size~$N$, meaning
$N$ times  
a factor $N^{o(1)}$,
i.e.,  
a factor $N^{\eta(N)}$ with $\lim_{N\to \infty} \eta(N)$ $=$ $0$.

We give algorithms that compute essentially optimal certificates for the positive semidefiniteness,
Frobenius form, characteristic and minimal polynomial of an $n\times n$ dense
integer matrix $A$.  Our certificates can be verified in Monte-Carlo bit
complexity $(n^2 \lognormA)^{1+o(1)}$, where $\lognormA$ is the bit size of the
integer entries,
solving  
an open problem in [Kaltofen, Nehring, Saunders, Proc.\ ISSAC 2011] subject
to computational hardness assumptions.

Second, we give algorithms that compute certificates for the rank of sparse or structured
$n\times n$ matrices over an abstract field, whose Monte Carlo verification complexity
is $2$ matrix-times-vector products $+$ $n^{1+o(1)}$ arithmetic operations in the field.
For example, if the $n\times n$ input matrix is sparse with
$n^{1+o(1)}$ non-zero entries,  
our rank certificate can be verified in
$n^{1+o(1)}$ field operations.
 
This extends also to integer matrices with only an extra $\lognormA^{1+o(1)}$
factor.

All our certificates are based on interactive verification protocols with
the interaction removed by a Fiat-Shamir identification heuristic.  The
validity of our verification procedure is subject to
standard computational hardness assumptions from cryptography.  
\par  
\end{abstract}
 
{\small 
  \category{I.1.2 }{Computing Methodologies}{Symbolic and Algebraic Manipulation}
  \terms{Theory, algorithms, verification}
}

\section{Introduction}
Suppose you want to externalize your computations to
cloud services. Prior to payment of the services, it would be desirable
to verify that the returned result has been correctly computed by the cloud
servers. This model is economically viable only if the verification
process requires less resources than the computation itself.
It is therefore important to design {\em certificates} that can be
used to verify a result at a lower cost than that of recomputing it.

For instance, a primality certificate is a formal proof that a number
is prime. In \cite{Pratt:1975:primroot}, primality is assessed by
presenting a primitive root and the factorization of $m-1$.
The latter can be checked fast by remultiplying, and then primitivity
is polynomially checkable.

In linear algebra our original motivation is related to \sos.
By Artin's solution to Hilbert 17th Problem, any polynomial inequality
$\forall \xi_1,\ldots,\xi_n\in\R, f(\xi_1,\ldots,\xi_n)\geq
g(\xi_1,\ldots,\xi_n)$ can be proved by a fraction of \sos: 
\begin{equation}\label{eq:sos}
  \exists u_i, v_j \in \R[x_1,\ldots,x_n],
  f-g=\left(\sum_{i=1}^\ell u_i^2\right)/\left(\sum_{j=1}^m v_j^2\right)
\end{equation}
Such proofs can be used to establish global minimality for
 
$g = \inf_{\xi_v\in\R} f(\xi_1,\ldots,\xi_n)$ 
and constitute certificates in non-linear global optimization. 
A symmetric integer matrix $W\in \symm{\ZZ}{n}$ is \psd,
denoted by $W \succeq 0$, if all its eigenvalues, which then must be
real numbers, are non-negative. Then, a certificate for \psdness of rational
matrices constitutes, by its Cholesky factorizability, the final step
in an exact rational \sos proof, namely
\begin{multline}\label{eq:sdp} 
  \hspace*{-0.5em}\exists e\ge 0,\ W^{[1]}\succeq 0,\ W^{[2]} \succeq 0,\ W^{[2]}\neq \mathbf{0}:
  \\
  (f-g)(x_1,\ldots,x_n)\cdot(m_e(x_1,\ldots,x_n)^T W^{[2]}m_e(x_1,\ldots,x_n))=\\
  m_d(x_1,\ldots,x_n)^T W^{[1]}m_d(x_1,\ldots,x_n),
\end{multline}
where the entries in the vectors $m_d,m_e$ are the terms occurring in $u_i,v_j$
in
(\ref{eq:sos}).
In fact, (\ref{eq:sdp}) is the \sdp that one solves.

Thus arose the question how to give possibly probabilistically checkable
certificates for linear algebra problems.
In \cite{Kaltofen:2011:quadcert} the certificates are restricted to those that
are checkable in essentially optimal time, that is, in bit complexity 
$(n^2 \lognormW)^{1+o(1)}$, 
where $\lognormW$ is the bit size of the entries in $W${}.
Quadratic time is feasible because a matrix multiplication $AB$ can be certified
by the resulting product matrix $C$ via Rusin Freivalds's
\cite{Freivalds:1979:certif} (see also \cite{Kimbrel:1993:PAV}) probabilistic
check: check $A(B v) = C v$ for a random vector $v${}.

Note that programs that check their results from \cite{Blum:1995:checkwork}
have the higher matrix-multiplication time complexity. 
In \cite{Kaltofen:2011:quadcert} a certificate for matrix
rank was presented, based on Storjohann's Las Vegas rank algorithm
\cite{Storjohann:2009:IMR}, but matrix \psdness remained open.  
Also the presented certificate for the rank did not take into account a possible
structure in the matrix.

In the following we solve these two problems.
In both cases, \psdness and structured or blackbox matrices, our solution is to
use either {\em interactive} certificates under the random oracle model, or
heuristics under standard computational hardness assumptions from cryptography. 
Removing the cryptographic assumptions remains however a fundamental open
problem. Providing certificates to other problems, such as the determinant or
the minimal and characteristic polynomial of blackbox matrices, is also open.

In Section~\ref{sec:notions},
we detail the different notions of certification that can be used and in
particular the relaxation we make over the certificates
of~\cite{Kaltofen:2011:quadcert}: in the certificates presented here, we allow
the verifier to provide the random bits used by the prover, in an interactive
manner.
We also present in this section the Fiat-Shamir derandomization heuristic that
can turn any interactive certificate into a non-interactive heuristic
certificate.
 
More precisely, the idea is to devise an interactive protocol for the random
oracle model, and then to replace oracle accesses by the computation of an
appropriately chosen function
$h$~\cite{Fiat:1986:Shamir,Bellare:1993:randomoracle}.

Then we first present in Section~\ref{sec:ints} an interactive certificate for
the Frobenius normal form that can be verified in
$\bigO{n^{2+o(1)}(\lognormA)^{1+o(1)}}$ binary operations, as in
\cite{Kaltofen:2011:quadcert}, but our new certificate also occupies an optimal
space of $\bigO{n^{2+o(1)}(\lognormA)^{1+o(1)}}$ bits. 
This is an order of magnitude improvement over
\cite[Theorem~4]{Kaltofen:2011:quadcert}.  
This certificate can then be used as in the latter paper to certify the minimal
and characteristic polynomial as well as \psdness, while keeping the lower
memory requirements. In the same section we also present another, stand-alone,
characteristic polynomial certificate, which can also be used for \psdness, with
slightly smaller random evaluation points.

Finally in Section~\ref{sec:sparserank} we present a new certificate for the
rank of sparse or structured matrices.
 
The certificate combines an interactive certificate of non-singularity,
giving a lower bound to the rank, with an interactive certificate for an upper
bound to the rank.
Overall the interactive certificate for the rank requires only
$2\Omega+n^{1+o(1)}$ arithmetic operations over any coefficient domain,
where $\Omega$ is the number of operations required to perform one
matrix-times-vector product. 
 
This certificate is then extended to work over the integers with only an extra $\lognormA^{1+o(1)}$ factor.
For instance, if the matrix is sparse with only $n^{1+o(1)}$
non-zero elements, then the certificate verification is essentially linear.

\section{Notions of certificate}\label{sec:notions}
The ideas in this paper arise from linear algebra, probabilistic algorithms,
program testing and cryptography. 

We will in particular combine:
\begin{itemize}
\item the notions of
certificates for linear algebra of Kaltofen et al.~\cite{Kaltofen:2011:quadcert},
which extend the randomized algorithms of Freivalds~\cite{Freivalds:1979:certif},
 
and reduce the computation cost of the program checkers of
Blum and Kannan~\cite{Blum:1995:checkwork},

\item with probabilistic interactive proofs of
Babai~\cite{Babai:1985:arthurmerlin} and Goldwasser et
al.~\cite{Goldwasser:1985:IPclass},
\item  as well as Fiat-Shamir
heuristic~\cite{Fiat:1986:Shamir,Bellare:1993:randomoracle} turning interactive
certificates into non-interactive heuristics subject to computational hardness.
\end{itemize}

We first recall some of these notions and then define
in Section~\ref{ssec:intcert} what we mean by perfectly complete, sound and
efficient interactive certificates.

\subsection{Arthur-Merlin interactive proof systems}
 
A proof system usually has two parts, a theorem $T$ and a proof $\Pi$, and the
validity of the proof can be checked by a verifier~$V$.
Now, an {\em interactive proof}, or a
{\em $\sum$-protocol}, is a dialogue between a prover $P$ (or {\em Peggy} in
the following) and a verifier $V$ (or {\em Victor} in the following), where $V$
can ask a series of questions, or challenges, $q_1$, $q_2$, $\ldots$ and $P$ can
respond alternatively with a series of strings $\pi_1$, $\pi_2$, $\ldots$, the
responses,  in order to prove the theorem $T$. 
The theorem is sometimes decomposed into two parts, the hypothesis, or input,
$H$, and the commitment, $C$. Then the verifier can accept or reject the
proof: 
$V (H,C, q_1, \pi_1,q_2,\pi_2,\ldots)\in\{\text{accept}, \text{reject}\}$. 

To be useful, such proof systems should satisfy {\em completeness} (the prover
can convince the verifier that a true statement is indeed true) and soundness
(the prover cannot convince the verifier that a false statement is true). 
More precisely, the protocol is {\em complete} if the probability that a true
statement is rejected by the verifier can be made arbitrarily small.
Similarly, the protocol is {\em sound} if the probability that a false statement
is accepted by the verifier can be made arbitrarily small.
The completeness (resp. soundness) is {\em perfect} if accepted (resp. rejected)
statement are always true (resp. false). 

It turns out that interactive proofs with perfect completeness are as powerful
as interactive proofs \cite{Furer:1989:perfectcomp}.
Thus in the following, as we want to prove correctness of a result more than
proving knowledge of it, we will only use interactive proofs with perfect
completeness. 

On the one hand, if a protocol is both perfectly complete and perfectly sound
then it is deterministic. On the other hand, if at least one of completeness
and soundness is not perfect, then the proof is probabilistic and correspond to
Monte Carlo algorithms (always fast, probably correct).

After submitting our paper to ISSAC, we learned of related work on
certifying the evaluation of a Boolean or arithmetic circuit 
by multi-round interactive protocols
GKR'08 \cite{Goldwasser:2008:delegating} and Thaler'13
\cite{Thaler:2013:crypto}, 
where the verifier runs in  
essentially linear time in the input size
and the prover is limited in compute power.
Since our computational problems in linear algebra
can be solved by a Boolean or arithmetic circuit of size $S(n)$,
where $S(n)$ is the bit or arithmetic computation time, 
the GKR'08 and Thaler'13 protocols constitute an alternate (implicit) approach
to certification of the output.  Those protocols (for multiple rounds) have
smaller communication cost, namely, linear in the depth of the circuit, 
than ours, which can be linear or quadratic in the matrix dimension.
But GKR'08 and Thaler'13 may require a more powerful prover:
their provers may require
a factor of $O(\log(S(n)))$ or a constant factor more time than $S(n)$.
Several of our certificates
are computed faster, in $S(n) + o(S(n))$ time, where $S(n)$ is the sequential time,
e.g., our certificate for \psdness.
A technical condition for the GKR'08 and Thaler'13 verifier
is that its complexity depends linearly on the depth of the circuit, 
so an algorithm of depth $\Omega(n^{2.1})$ would not satisfy our requirement
of linear complexity for the verifier,  
but none of our algorithms have that much depth. 

More significantly, our protocols and certificates are explicit:
the verifier need not know the circuit for the computation,
nor compute certain properties of it.
Nevertheless, GKR'08 \cite{Goldwasser:2008:delegating} and Thaler'13  \cite{Thaler:2013:crypto}
guarantee multi-round protocols,  
and motivate the search for explicit certificates.

In our explicit setting, we {\em require} that the verifier algorithm has
lower computational complexity than any known algorithm computing the property.
 
Moreover, we want  
the complexity of the prover algorithm  
as close as
possible to the best known algorithm computing the property without
certificates.  

\subsection{Certificates in linear algebra}
For Blum and Kannan~\cite{Blum:1995:checkwork} a program checker for a
program~$P$ is itself a program~$C$. For any instance~$I$ on which program~$P$
is run, $C$ is run subsequently. $C$ either certifies that the program~$P$ is
correct on~$I$ or declares~$P$ to be buggy. There, the programs can be rerun on
modified inputs, as in their matrix rank check, and thus might require more time
to check their work than to do the work itself.

On the contrary, in \cite{KLYZ09,Kaltofen:2011:quadcert}, a certificate for a
problem that is given by input/output specifications is an input-dependent data
structure and an algorithm that computes from that input and its certificate the
specified output, and that has lower computational complexity than any known
algorithm that does the same when only receiving the input. Correctness of the
data structure is not assumed but validated by the algorithm. 
With respect to interactive proofs, the input/output is related to the property
to be proven together with the commitment.
However, as no interaction is possible between the prover and the verifier, this
amounts to using a single round protocol where the prover sends only the
commitment and then the verifier accepts it or not. 

In this paper we use a modified version where we allow interactive exchanges
between the prover and the verifier but preserve the requirements on lower total
complexity for the verifier.
Moreover, we then can convert back these two-rounds protocol into one round
protocols via Fiat-Shamir heuristic: hash the input and commitment with an
unpredictable and universal hash function  (such as a cryptographic hash
function), to simulate the random challenges proposed by the verifier. 

It turns out that it seems easier to design certificates that are interactive
than to design directly single round certificates. This could be related to the
power of the interactive proof system complexity class (IP) and the
probabilistically checkable proofs (PCP).

\subsection{Interactive certificates}\label{ssec:intcert}
 
There exists ways to reduce interactive proofs with $k$ rounds to interactive
proofs with perfect completeness and $2$ rounds,
by increasing the verifier's complexity by time exponential in $k$.
Here we will limit ourselves to $2$ rounds
for our definition of interactive certificates. 

More precisely, in the following we use interactive certificates of a given
property, mainly as  {\em two-rounds probabilistic $\sum$-protocols with perfect completeness}: 
\begin{enumerate}
\item The prover of a property sends a commitment to the verifier.
\item The verifier sends back a (randomly sampled) challenge, potentially
  depending on both the property and the commitment.
\item The prover completes the protocol with a response convincing
  the verifier of the property.
\end{enumerate}
In order to become an {\em interactive certificate}, this two round
$\sum$-protocol should then satisfy soundness, perfect completeness and
efficiency as follows: 
\begin{enumerate}\renewcommand {\theenumi}{\roman{enumi}}
\item The protocol is {\em perfectly complete}: a true statement will always be
  accepted by the verifier.
\item The protocol is {\em sound}: the probability that a false statement will
  be accepted by the verifier can be made arbitrarily small.
\item The protocol is {\em efficient}: the verifier has lower computational
  complexity than any known algorithm that computes the true statement when only receiving the input.
\end{enumerate}

The interactive certificate can also be said to be {\em essentially optimal}
when the verifier needs only time and space complexity of the same order of
magnitude as the size of the input and output to verify the latter.

With this relaxed model, we are able in the following to improve on
some space complexities for integer linear algebra problems and also
on time complexities for some problems over generic domains, like the rank of
blackbox matrices.

\subsection{Fiat-Shamir derandomization into a single
  heuristic round}\label{ssec:fiatshamir}
In a practical perspective (say when using a compiled library, rather than an
interpreter; or when posting the certificate in question) it is not always
possible for the verifier (a user wanting a result) to interact with the prover
(the program). 

Then, there is always the possibility to transform an interactive certificate
into a non-interactive heuristic. 
Here we use the strong Fiat-Shamir
heuristic~\cite{Fiat:1986:Shamir,Bellare:1993:randomoracle,Bernhard:2012:fiatshamir},
where the random challenge message of the verifier is replaced by a
cryptographic hash of the property {\em and} the commitment message.
In practice, the cryptographic hash can be used as a seed for a pseudo-randomly
generated sequence that the prover can generate a priori. 
For an a posteriori verification, the verifier decides whether to accept or not
the certificate, as in two rounds interactive protocols, but has also to check
that the challenge used by the prover has really been generated using the input
and commitment as seeds.

In this setting, breaking the protocol is somewhat equivalent to breaking the
cryptographic hash function: finding a combination of input and false commitment
that will be accepted by the verifier relates to knowing in advance some parts
of the output of the hash function.
See for instance Section~\ref{ssec:bbs} where
breaking the protocol is equivalent to predicting the value of some bits in a
hash, and that can for instance be used to factor integers if Blumb-Blum-Shub
hash function is used. 

Note that it is important to use the {\em strong} heuristic that uses a
combination of both the input and the commitment for the hashing.
See for instance Section~\ref{ssec:charpoly} where we need the result itself to
be part of the seed in order to obtain a correct certificate.

\section{Reducing space with respect to non-interactive certificates{\customlinebreak} over the
  integers}\label{sec:ints}
\subsection{Interactive certificate for residue systems}
In \cite[Theorem~5]{Kaltofen:2011:quadcert}, the given certificates for the rank
and determinant of an integer matrix are essentially optimal whereas the
certificates for the Frobenius normal form (without transformation matrices),
the characteristic and minimal polynomial and \psdness are not: they require
residue systems that occupies cubic bit space whereas the input and results
occupy only a quadratic number of bits.

Those residue systems allow the verifier to check an integer matrix
factorization 
($A=LU$ for Gaussian elimination or $A=SFS^{-1}$ for the Frobenius form) where
the resulting factors are in general of cubic size (quadratic number of entries
but each one with linear magnitude) via Freivalds' certificate.
The trick is to store these factorizations modulo many distinct primes.
Then if the integer matrix factorization is not correct it means that $A-LU$
(resp. $A-SFS^{-1}$) is non zero. 
Therefore, from a bound on the maximal possible size of this difference (roughly
cubic), it cannot be zero modulo a large number of primes. 
Consequently, if the set
of distinct primes is larger than the bound, selecting a random prime $p$ in the
set and checking whether $A-LU$ (resp. $A-SFS^{-1}$) is zero modulo $p$ would
reveal the false statement with high probability.

Our idea here is to use several rounds interactive certificates: instead of
storing the factorizations modulo many distinct primes, just compute them on
demand of the verifier. The verifier has just to select random primes and the
prover will respond with the factorization modulo these primes.

\begin{theorem}\label{thm:RNS}
  Let $A\in \Z^{n\times n}$. There exists an {\em interactive certificate} for the
  Frobenius normal form, the characteristic or minimal polynomial of $A$. The
  interactive certificate can be verified in $n^{2+o(1)}(\lognormA)^{1+o(1)}$
 
  bit operations
  and occupies $n^{2+o(1)}(\lognormA)^{1+o(1)}$ bit space.
\end{theorem}
\begin{proof}
  Use the same algorithm as in~\cite[Theorem 4]{Kaltofen:2011:quadcert} but
  replacing the random choice by the verifier of a given tuple $(p,S_p,F_p,T_p)$
  (where $T_p\equiv S_p^{-1} \mod p$) by the choice of a random prime $p$ by the
  verifier  and a response of a corresponding $(S_p,F_p,T_p)$ modulo $p$ by the
  prover. 
\end{proof}

\begin{corollary} There exists a {\em non-interactive heuristic certificate} for
  the Frobenius normal form, the characteristic or minimal polynomial that
  occupy the same space and can be verified in the same time.
\end{corollary}
\begin{proof} We use Fiat-Shamir. The prover:
  \begin{enumerate}
  \item
  computes the integer Frobenius normal form $F$ (or the characteristic or
  minimal polynomial) over the integers;
  \item 
  then he chooses a cryptographic hash function and a pseudo-random prime 
  generator; 
  \item 
  he computes the hash of the input matrix together with the result;
  \item 
  this hash is used as a seed for the pseudo-random prime generator to
  generate one (or a constant number of) prime number(s);
  \item 
  the prover finally produces the Frobenius normal form and the change of basis
  modulo that prime(s).
  \end{enumerate}
  The certificate is then composed of the input, the output, the hash function,
  the pseudo-random prime generator, the generated prime numbers and the
  associated triples $(S_p,F_p,T_p)$.

  The verifier then:
  \begin{enumerate}
  \item checks that the hash function and the pseudo-random prime generator
    are well-known, cryptographically secure, functions;
  \item checks that he can recover the primes via hashing the combination of the
    input and the output;
  \item and verifies the zero equivalence modulo $p$ of $(F-F_p)\mod p$,
    $(S_pT_p-I)\mod p$ and $(S_p F_p T_p-A)\mod p$.
  \end{enumerate}
\end{proof}

\subsection{Direct interactive certificate for the
  characteristic polynomial and positive definiteness of integer
  matrices}\label{ssec:charpoly}
In~\cite{Kaltofen:2011:quadcert}, the certificate for characteristic
polynomial occupies roughly $n^{3+o(1)}$ bit space as it requires the
Frobenius matrix normal form with a similarity residue system with primes
bounded by $\bigO{ n(\log(n)+\lognormA) }$.

As shown in Theorem~\ref{thm:RNS}, with an interactive certificate and a random
oracle for the choice of prime numbers of the latter size, this yields an
interactive certificate with only $n^{2+o(1)}$ bit space requirements.
 
We propose in the following Figure~\ref{fig:charp} a simpler
certificate, still relying on the determinant certificate, but with
evaluation points bounded only by $\bigO{n}$. This gives a similar but smaller
$o(1)$ factor in the complexity.

\begin{figure*}[htbp]\center 
  \noindent\resizebox{.8\linewidth}{!}{$$
    \xy
    \xymatrix@R=7pt@C=1pc@W=15pt{
      &Peggy &&&& Victor\\
      Input&&& A\in\Z^{n\times n} &&\\
      \ar@{.}[rrrrrr]&&&&&&\\
      Commitment&g\in\Z[X]=charpoly_A&\ar[rr]^*[*1.]{1:\ g(X)} &&& degree(g)\stackrel{?}{=}n\\
      Challenge& & && \ar[ll]_*[*1.]{2:\ \lambda}& \lambda\in\Z\\
      Response& \delta\in\Z=det(\lambda I-A)&\ar[rr]^*[*1.]{3:\ \delta}&&& \delta\stackrel{?}{=}g(\lambda)\\
      & C: Cert(\delta=det(\lambda I-A))&\ar[rr]^*[*1.]{4:\ C}&&& \delta\stackrel{?}{=}det(\lambda I-A)\\
    }\POS*\frm{-}
    \endxy
    $$
  }
  \caption{Interactive certificate for the characteristic polynomial}\label{fig:charp}\customvspace{5pt}
\end{figure*}
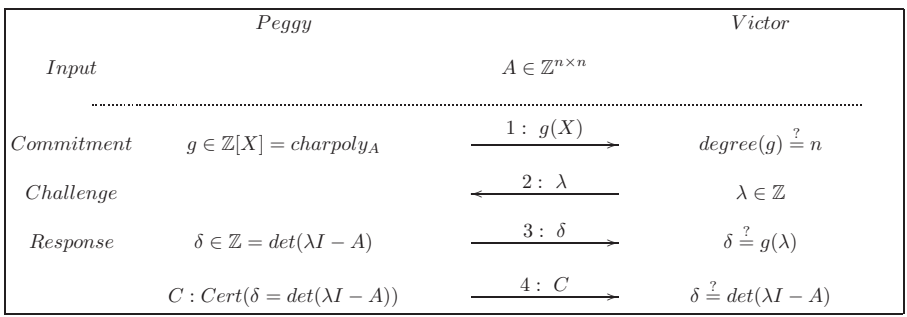

\begin{theorem}\label{th:charp}
  For $A\in\Z^{n\times n}$, 
  the interactive certificate of Figure~\ref{fig:charp} for
  the characteristic polynomial is sound, perfectly complete and the number of
  operations performed by the verifier, as well as the bit space
  required to store this certificate, is bounded by
  $n^{2+o(1)}(\lognormA)^{1+o(1)}$.
\end{theorem}
\begin{proof}
  For the determinant certificate we
  use~\cite[Theorem~5]{Kaltofen:2011:quadcert} whose complexity
  matches that of the present theorem.\\
   
  If Peggy is honest then the definition of the
  characteristic polynomial yields $charpoly_A=det(X I -A)$ and thus
  the protocol is perfectly complete.\\
   
  If Peggy is dishonest then $g-charpoly_A$ being of degree at most
  $n$, it has at most $n$ roots. Thus if Victor samples random
  elements among the first say $c n$ integers, {\em after} the
  commitment $g$, the probability that
  Victor accepts the certificate is less than $1/c$. If the protocol
  is repeated $k$ times with independent draws of $\lambda$, then the
  probability that Victor accepts it $k$ times is lower than
  $\left(\frac{1}{c}\right)^k$ and therefore the protocol is sound.\\

  For the complexity, one chooses a constant $c>2$ so that $\lambda$
  has $\bigO{\log(n)}$ bits. Thus $\delta$, as the determinant of 
  $\lambda I -A$, is bounded by Hadamard's bound to
  $\bigO{n\log(\| A\mspace{1mu} \|+n)}$ bits.  
  With Horner evaluation and Chinese
  remaindering, the check 
  $g(\lambda)\stackrel{?}{=}\delta$ can thus be performed in less than
  $\bigO{n^2\log(\| A\mspace{1mu} \|+n)}$ operations.  
  This is within the announced bound. 
\end{proof}

\begin{corollary}
  Let A be an $n\times n$ symmetric matrix having
  minors bound $H_A$ of bit length $\log_2(H_A)=n^{1+o(1)}$. 
  The signature of A can be verified by an interactive certificate in 
  $n^{2+o(1)}$ binary operations with a 
  $n^{2+o(1)}$ bit space characteristic polynomial certificate. 
  Thus the same certificate serves for positive or negative definiteness
  or semidefiniteness.
\end{corollary}
\begin{proof}
  We just use the certificate of
  \cite[Corollary~1]{Kaltofen:2011:quadcert} but replace their
  characteristic polynomial certificate by the interactive one of
  Figure~\ref{fig:charp} and Theorem~\ref{th:charp}.
\end{proof}

\section{Interactive certificate for{\customlinebreak} the rank of sparse matrices}\label{sec:sparserank}
 
Now we turn to matrices over any domain and count arithmetic
operations instead of bit complexity. That is to say we consider that the four
arithmetic operations over the domain, plus say equality testing, random
sampling of one element, etc., are all counted as $1$ operation.
 
For the sake of simplicity we will use the notation $\F$ as for finite fields
but the results are valid over any abstract field, provided that the random
sampling is done on a finite subset $S$ of the domain. 

We improve on $\bigO{n^2}$ certificates for the rank (given
with say an $LU$ factorization), when the matrix is sparse, structured or given
as blackbox. That is to say when the product of the matrix by a vector 
requires strictly less arithmetic operations than what would be required if the
matrix was dense. 
 
If the matrix is given as a blackbox, then the only possible operation with the
matrix is the latter matrix-times-vector product.
 
For a matrix of rank $r$, if $\Omega$ is the cost of one of those
matrix-times-vector product, the blackbox certificates
of~\cite{Saunders:2004:rankcert} would also require $\bigO{nr}$ extra arithmetic
operations and at least $\bigO{r}$ extra matrix-times-vector products for a
total of $\bigO{r\Omega+nr}$ arithmetic operations.

In the following we show that it is possible to reduce the time and space
complexity bounds of verifying certificates for the rank of blackbox matrices to
only $2\Omega+n^{1+o(1)}$ arithmetic operations. 
This is essentially optimal, e.g. for sparse matrices, as reading and
storing a matrix of dimensions $n\times n$ should also require
$\bigO{\Omega+n}$ operations. 
 
We then extend this result, but in bit complexity, to also certify the rank of
integer matrices with essentially the same optimal complexity bounds.

We proceed in two steps. First we certify that there exists an
$r\times r$ non-singular minor in the matrix. Second, we precondition the matrix
so that it is of generic rank profile and exhibit a vector in the null-space of
the leading $(r+1)\times(r+1)$ minor of the preconditioned matrix. 
 
\subsection{Certifying non-singularity}

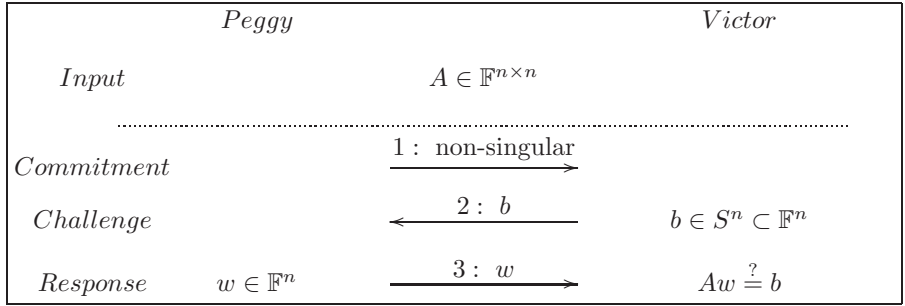
\begin{figure*}[htbp]\center 
  \noindent\resizebox{.65\linewidth}{!}{$$
    \xy
    \xymatrix@R=7pt@C=1pc@W=15pt{
      &Peggy &&&& Victor\\
      Input&&& A\in\F^{n\times n} &&\\
      \ar@{.}[rrrrrr]&&&&&&\\
      Commitment& & \ar[rr]^*[*1.]{1:\ \text{non-singular}} &&&\\
      Challenge& & && \ar[ll]_*[*1.]{2:\ b}& b\in S^n\subset\F^n\\
      Response& w\in\F^n&\ar[rr]^*[*1.]{3:\ w}&&& Aw\stackrel{?}{=}b \\
    }\POS*\frm{-}
    \endxy
    $$
  }
  \caption{Blackbox interactive certificate of non-singularity}\label{fig:nonsing} 
\end{figure*}

\begin{theorem}\label{thm:nonsing}
  Let $S$ be a finite subset of $\F$ with at least two distinct elements.
  For $A\in\F^{n\times n}$, whose matrix-times-vector products costs $\Omega$
  operations in $\F$, the interactive certificate of Figure~\ref{fig:nonsing}
  for non-singularity is sound, perfectly complete and the number of arithmetic
  operations performed by the verifier is bounded by $\Omega+n^{1+o(1)}$.

\end{theorem}
\begin{proof}
  If Peggy is honest, then she can solve the system with an
  invertible matrix and provide $w=A^{-1}b$ to Victor. Therefore the protocol is
  perfectly complete.\\  
   
  If Peggy is dishonest, then it means that $A$ is singular. 
  Therefore, it means that the rank of $A$ is at most $n-1$.
   
  We use, e.g., Gaussian elimination to get $A=PLUQ$, 
  where $P$ and $Q$ are permutation matrices, $L$ is unit invertible lower
  triangular and $U$ is upper triangular. As the rank of $A$ is at most $n-1$,
  $U$ is of the form
  $\left[\begin{smallmatrix}U_1&U_2\\0&0\end{smallmatrix}\right]$ where
  $U_1\in\F^{(n-1)\times (n-1)}$ is upper triangular.
  One then sees
  that making the system inconsistent is equivalent to setting to
  zero at least the last entry of $L^{-1}P^{-1}b$ in $\F^n$.

  Thus, with probability at least $1-1/\card{S}$, the challenge vector
  proposed by Victor makes the system inconsistent\footnote{Alternatively, one
    can get the same result with the Schwartz-Zippel lemma applied to the linear
    function whose kernel is the nullspace of $A$ as in
    e.g.~\cite[Theorem~2.2]{Giesbrecht:1998:CIS}.}. 
  In the latter case, Peggy will never be able to find a solution to the
  system. 

  Thus Victor can accept the certificate of Peggy only when he
  has randomly found a consistent vector. 
  The probability that this happens $k$ times with $k$ independent
  selections of $b$ is bounded by $\frac{1}{\card{S}^k}$.
  Therefore, when the matrix is singular, Victor can accept repeated
  applications of the protocol only with negligible probability and the
  protocol is sound.\\

  For the complexity, Victor needs to perform one matrix-times-vector product
  with $A$, of arithmetic complexity $\Omega$. 
  Victor also needs to produce a random vector of size $n$ of elements
  in $S$ and perform a vector equality comparison.

\end{proof}

\subsection{Certifying an upper bound for the rank}
For an upper bound, we precondition $A\in\F^{m\times n}$ of rank~$r$ so that the
leading $r\times r$ minor of the preconditioned matrix is non-zero and then
present a non-zero vector in the nullspace of the $r+1$ leading minor. 
We use the butterfly probabilistic preconditioners of 
\cite[Theorem 6.3]{Chen:2002:EMP} that can precondition an $n\times n$ matrix of
rank $r$ so that the first $r$ rows of the preconditioned matrix become linearly
independent with high probability. 
We denote by $\B_S^{n\times n}$ the set of such butterfly networks composed by
less than $n(\log_2(n))$ switches of the form
$\left[\begin{array}{cc}1&\alpha\\1&1+\alpha\end{array}\right]$, for
$\alpha\in S\subset\F$. Choosing a random butterfly reduces to choosing an element
$\alpha$, a row index and a column index for each of its switches.

\cite[Theorem 6.3]{Chen:2002:EMP} works for square matrices but can easily be
extended to work for rectangular matrices as follows.
\begin{lemma}\label{lem:genbut}
  Let $\F$ be a field and $S$ be a finite subset of $\F$.
  Let $A$ be an $m\times n$ matrix over $\F$
  with $r$ linearly independent rows. Then a butterfly preconditioner
  $U\in\B_S^{m\times m}$ will make the first $r$ rows of $UA$ linearly independent
  with probability not less than $1-\frac{r\lceil \log_2(m)\rceil}{|S|}$.
\end{lemma}
\begin{proof}
From  
\cite[Theorem 6.2]{Chen:2002:EMP}, 
we know that a depth $\lceil\log_2(m)\rceil$ butterfly network can switch any
$r\leq m$ indices into the continuous block $1,2,\ldots,r$. 
The proof of \cite[Theorem 6.3]{Chen:2002:EMP} uses the butterfly to
permute and combine the rows of $A$ independently of its number of columns.
Therefore $UA$ will have its first $r$ rows linearly independent with probability not less than $1-\frac{r\lceil \log_2(m)\rceil}{|S|}$.
\end{proof}
We can thus now turn to our interactive certificate of an upper bound to the rank.

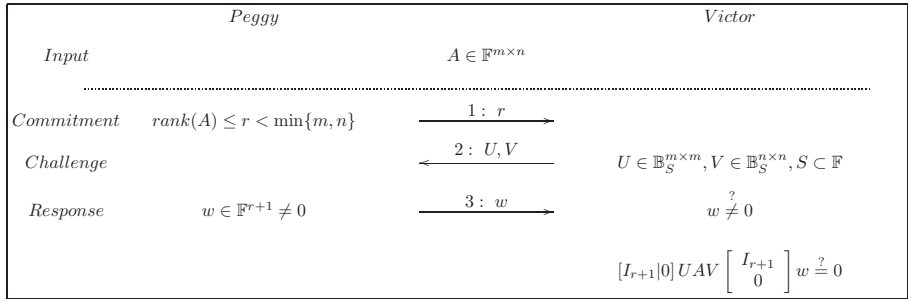
\begin{figure*}[htbp]\center 
  \noindent\resizebox{0.9\linewidth}{!}{$$
    \xy
    \xymatrix@R=7pt@C=1pc@W=15pt{
      &Peggy &&&& Victor\\
      Input&&& A\in\F^{m\times n} &&\\
      \ar@{.}[rrrrrr]&&&&&&\\
      Commitment&rank(A) \leq r < \min\{m,n\} &\ar[rr]^*[*1.]{1:\ r} &&\\
      Challenge& & & & \ar[ll]_*[*1.]{2:\ U,V}&
      U\in\B_S^{m\times m},V\in\B_S^{n\times n}, S\subset\F\\
      Response& 
      w\in\F^{r+1}\neq 0&\ar[rr]^*[*1.]{3:\ w}&&& 
      w \stackrel{?}{\neq}0\\
      &&&&&\left[ I_{r+1}|0 \right] U A V \ensuremath{\left[\begin{array}{cc}I_{r+1}\\0\end{array} \right]} w \stackrel{?}{=}0
      \\
    }\POS*\frm{-}
    \endxy
    $$
  }
  \caption{Blackbox interactive certificate for an upper bound to the
    rank}\label{fig:upperbound} 
\end{figure*}

\begin{theorem}\label{thm:upper}  
  Let $A\in\F^{m\times n}$, whose matrix-times-vector products costs $\Omega$
  operations in $\F$ and let $S$ be a finite subset of $\F$ with 
  $\card{S}> 2\min\{m,n\}(\lceil\log_2(m)\rceil+\lceil\log_2(n)\rceil)$.  
  The interactive certificate of Figure~\ref{fig:upperbound} proving an
  upper bound for the rank of $A$ is sound, perfectly complete and the number of
  arithmetic operations performed by the verifier is bounded by

  $\Omega+(m+n)^{1+o(1)}$.
\end{theorem}
\begin{proof}
  If Peggy is honest this means that the rank of~$A$ is upper bounded by
  $r<\min\{m,n\}$. Thus the rank of 
  $M=\left[ I_{r+1}|0 \right] U A V \left[\begin{array}{cc}I_{r+1}\\0\end{array}
  \right]\in\F^{(r+1)\times (r+1)}$ is also upper bounded by~$r$. 
  Therefore, there exist at least one non-zero vector~$w$ in the nullspace of
  $M$. Hence Peggy can produce it and the protocol is perfectly complete.\\
   
  If Peggy is dishonest, this means that the rank of $A$ is at least $r+1$.
   
  Now, from Lemma~\ref{lem:genbut}, 
  the butterfly preconditioner $U\in\B_S^{m\times m}$ will make the first $r+1$
  rows of $UA$ linearly dependent with probability less than
  $\frac{(r+1)\lceil\log_2(m)\rceil}{\card{S}}$. Similarly the butterfly preconditioner
  $V\in\B_S^{n\times n}$ will make the first $r+1$ columns of $AV$ linearly
  dependent with probability less  than 
  $\frac{(r+1)\lceil\log_2(n)\rceil}{\card{S}}$. 
  Overall the $(r+1)\times(r+1)$ leading principal minor of $UAV$ will be non-zero
  with probability at least
  $1-\frac{(r+1)(\lceil\log_2(m)\rceil+\lceil\log_2(n)\rceil)}{\card{S}}$.
  The latter is greater than
  $1-\frac{\min\{m,n\}(\lceil\log_2(m)\rceil+\lceil\log_2(n)\rceil)}{\card{S}}\geq \frac{1}{2}$.
  In this case the minor is
  invertible and Peggy will never be able to produce a non-zero vector in its
  kernel. The only possibility for Victor to accept the certificate is thus that
  the leading minor is zero and the probability that this happens $k$ times with
  $k$ independent selections of $U$ and $V$ is thus bounded by $\frac{1}{2^k}$. 
  Thus Victor can accept repeated applications of the protocol only with
  negligible probability and the protocol is sound.\\
  
  For the complexity, \cite[Theorem 6.2]{Chen:2002:EMP} gives us that
  butterflies of respective sizes $m\lceil\log_2(m)\rceil/2$ and
  $n\lceil\log_2(n)\rceil/2$ are sufficient. 
   
  Victor thus needs to produce
  $(m\lceil\log_2(m)\rceil/2 + n\lceil\log_2(n)\rceil/2)$ random
  elements in $S$ and also $\bigO{\log_2(m)+\log_2(n)}$ random bits
  for the row and columns indices. 
  Then the successive applications of $U$, $A$ and $V$ to a vector cost no more
  than $3m\lceil\log_2(m)\rceil/2+\Omega+3n\lceil\log_2(n)\rceil/2$ arithmetic operations.
\end{proof}

\begin{remark}
{\upshape 
  Over small fields, 
   
  it might not be possible to find a sufficiently large subset $S$. 
  Then one can use extension fields or change the preconditioners.
  For instance, $\tilde{U}\in\W^{(r+1)\times m}$, and respectively
  $\tilde{V}\in\W^{n\times (r+1)}$, can be taken as sparse matrix preconditioners, as in
  \cite{Wiedemann:1986:SSLE} (see also \cite[Corollary 7.3]{Chen:2002:EMP}), and
  replace respectively $\left[ I_{r+1}|0 \right] U$ and $V
  \ensuremath{\left[\begin{array}{cc}I_{r+1}\\0\end{array} \right]}$.
  They are randomly sampled with the 
  Wiedemann distribution, denoted by $\W$, and have thus not more than
  $2n(2+\log_2(n))^2$ non zero entries with probability at least $1/8$,
  \cite[Theorem 1]{Wiedemann:1986:SSLE}. 
\qed} 
\end{remark}

\subsection{Blackbox interactive certificate for the{\customlinebreak} rank}  
Now we can propose a complete certificate for the rank. 
 
If the matrix is full rank, then it is sufficient to produce a certificate for
a maximal lower bound.
Otherwise, it will use a non-singularity certificate on a sub-matrix of
dimension $r\times r$ together with an upper
bound certificate: 
for a matrix $A\in\F^{m\times n}$,
\begin{enumerate}
\item 
   
  Peggy computes $I^{(r)}\in\F^{r\times m}$
  (resp. $J^{(r)}\in\F^{n\times r}$) a concatenation of $r$ distinct
    row canonical vectors (resp. $r$ distinct column canonical
    vectors) such that $I^{(r)} A J^{(r)}$ is a non-singular
    submatrix of $A$;
    and produces the non-singularity certificate of
  Figure~\ref{fig:nonsing} on the latter. This
  provides a certified lower bound for the rank of $A$.
\item Peggy provides the certificate for an upper bound $r$ of the rank of $A$
  of Figure~\ref{fig:upperbound}. 
\item With certified lower bound $r$ and upper bound $r$, the rank
  is certified.
\end{enumerate}
Using Theorems~\ref{thm:nonsing} and~\ref{thm:upper}, we have proven:
\begin{corollary}\label{cor:sparserank}
 
  Let $A\in\F^{m\times n}$, whose matrix-vector products 
  costs $\Omega$ operations in $\F$ and let $S$ be a finite subset of~$\F$ with 
  $\card{S}> 2\min\{m,n\}(\lceil\log_2(m)\rceil+\lceil\log_2(n)\rceil)$.  
  The above $\sum$-protocol provides an interactive certificate for the rank
  of~$A$.
  This interactive certificate is sound, perfectly complete and the number of
  arithmetic operations performed by the verifier is bounded by
   
  $2\Omega+(m+n)^{1+o(1)}$.
\end{corollary}

\subsection{Blackbox interactive certificate for the{\customlinebreak} rank of integer matrices}
This rank certificate can also be used for integer matrices, at roughly the same
cost: just use the sparse certificate modulo a randomly chosen prime $p$.

Let $H=\min\{\sqrt{n}^m ||A||_\infty^m,\sqrt{m}^n ||A||_\infty^n\}$ be
Hadamard's bound for the invariant factors of $A$. Let $h=\log_2(H)$, there
cannot be more than $h$ primes reducing the 
rank. Therefore, the protocol is as follows:
\begin{enumerate}
\item Peggy produces the rank $r$ of $A$.
\item Victor selects $p$ randomly from a set with, say, $c\cdot h$ primes, for a
  constant $c>2$ (if the rank is correct then the probability that the rank of
  $A$ will be reduced modulo $p$ is then less than $1/2$).
\item Peggy produces the sparse rank certificate of
  Corollary~\ref{cor:sparserank} over the field of integers modulo $p$, that is,
  two vectors $w_1\in \Zp^n$ and $w_2\in\Zp^{r+1}$.
\item Victor checks the certificate of the rank modulo $p$ (to apply $A$ in
  $\Zp$, Victor applies $A$ over $\Z$ and then reduces the resulting vector
  modulo $p$).
\end{enumerate}

\begin{theorem}
  Let $A\in\Z^{m\times n}$, whose matrix-times-vector products
  costs $\Omega$ arithmetic operations in $\Z$.\\
   
  Let {\small $\mu=\max\{\lognormA,\log(m+n)+\log(\frac{1}{2}\log(m+n)+\lognormA)\}$}.
  The above $\sum$-protocol provides an interactive certificate for the rank
  of~$A$.
  This interactive certificate is sound, perfectly complete and the number of
  bit operations performed by the verifier is bounded by 
  $\left(2\Omega+(m+n)^{1+o(1)}\right)\mu^{1+o(1)}$
\end{theorem}
\begin{proof}
From  
the prime number theorem, we know that the $h$-th prime number is 
$O(h \log(h))$. Therefore it is possible to sample $c \cdot h$ distinct primes of
magnitude bounded by $O(h \log(h))$. Then, the cost of fast arithmetic modulo
any of these primes can be bounded by
$\eta=O(\log(h)^{1+o(1)})$
 
bit operations.
Then $\mu^{1+o(1)}\geq \max\{\lognormA,\eta\}^{1+o(1)}$ and applying $A$ to a vector in
$\Zp$ costs no more than $\Omega\mu^{1+o(1)}$ binary operations. Reducing the
coefficients of the resulting vector modulo $p$ then costs
$\left(\log(m)+\lognormA+\eta\right)^{1+o(1)}$, that is not more than the
matrix-times-vector product. 
 
From  
Corollary~\ref{cor:sparserank}, the cost of the certificate becomes 
$2\Omega\mu^{1+o(1)}+(n+m)\eta^{1+o(1)}$.  

Completeness is ensured from Corollary~\ref{cor:sparserank} and soundness by the
fact that the number of primes reducing the rank is less than $h$. 
Indeed, if Victor then samples from $c\cdot h$ primes with $c>2$, he can accept
$k$ repeated wrong ranks with probability bounded by $\frac{1}{2^k}$.
\end{proof}
Note that if the matrix is sparse with $(m+n)^{1+o(1)}$ elements, then 
$\Omega$ is $(m+n)^{1+o(1)}$ and the overall cost of the rank interactive
certificate over the integers remains essentially linear in 
the input size  
$(m+n)\lognormA$.  

\subsection{Reducing breaking the random oracle to
  factorization}\label{ssec:bbs}
Now we look at the derandomization of the previous certificates using the strong
Fiat-Shamir heuristic, see Section~\ref{ssec:fiatshamir},
where the random challenge messages of the verifier are replaced by a
cryptographic hash of the property and the commitment messages.

First, it is proven in~\cite{Pointcheval:1996:provsecsign}
that this methodology always produces digital signature schemes that
are provably secure against chosen message attacks in the "Random
Oracle Model" -- when the hash function is modeled by a random
oracle. 
In other words, it is equivalent for a dishonest Peggy to e.g. produce
consistent systems for singular matrices or to break the random oracle.

Second, we can, e.g., use the Blumb-Blum-Shub perfect
random generator~\cite{Blum:1982:BBS}:
it transforms a seed $x_0$ (for us the matrix) into a bit string
$b_1,\ldots,b_k$ with $b_i=x_i \mod 2;~x_{i+1}=x_i^e\mod N$ for some
RSA public key $(e,N)$. 
It is for instance shown in \cite{Fischlin:1997:SSP} that
knowing a number of $b_i$'s polynomial in the size of $N$, say
a number $B_N=\bigO{(\log N)^\gamma}$  
bits, enables one to factor it. 
Now, we show next that forging a consistency certificate
$Ax\stackrel{?}{=}b$ is equivalent
to predicting the value of at least one bit of the random right-hand
side vector~$b$. 
\begin{lemma} 
  Forging the non-singularity certificate fixes at least one bit of the random
  right-hand side vector.
\end{lemma}
\begin{proof}
  A singular matrix $A\in\F^{n\times n}$ has rank at most $n-1$.
  Write it as
  $A=P\left[\begin{smallmatrix}0&0\\*&L\end{smallmatrix}\right]UQ$ where $P$ and $Q$
  are permutation matrices, $L\in\F^{(n-1)\times (n-1)}$ is lower
  triangular and $U$ is unit invertible upper triangular. Then if $Aw=b$, 
  it means that  
  $\left[\begin{smallmatrix}0&0\\*&L\end{smallmatrix}\right] z=P^{-1}b$ for
  $z=UQw$. Therefore, the first entry of $P^{-1}b$ must be zero.
\end{proof}

Therefore, we fix the RSA modulus $N$ and require as a certificate
that the consistency check is repeated $B_N$ times. 
When the protocol is repeated with Fiat-Shamir derandomization,
we use as successive random vectors $b$, the hash of

the input and the previous iteration.
If Peggy can find a matrix $A$ of dimension $n$ (polynomial in the size of $N$)
for which she can forge the $B_N$ repeated applications of the certificate, then
 
she can predict $B_N$ bits of the Blumb-Blum-Shub hashes in
$\bigO{B_N(\Omega+n^{1+o(1)})}$ operations. 
Thus Peggy can factor $N$ in polynomial time.

\section*{Acknowledgments} 
We thank Brice Boyer 
, Shafi Goldwasser 
, Cl\'ement Pernet 
, Jean-Louis Roch 
, Guy Rothblum 
, Justin Thaler  
and the referees  
for their helpful comments.

\bibliographystyle{acm}  
\bibliography{biblio}
\balancecolumns
\end{document}